\documentclass[a4paper,11pt]{article}

\usepackage{graphicx}
\usepackage[font=small]{caption}
\usepackage{hyperref}
\usepackage{xcolor}
\usepackage{amsmath,amsfonts,amssymb}
\usepackage{amsthm}
\usepackage{bbm}
\usepackage{authblk}
\usepackage{fullpage}

\newcommand{\keywords}[1]{\noindent\textbf{\textit{Keywords---}} #1}

\newcommand{\N}{\ensuremath{\mathbb{N}}}
\newcommand{\Z}{\ensuremath{\mathbb{Z}}}
\newcommand{\R}{\ensuremath{\mathbb{R}}}
\newcommand{\F}{\ensuremath{{\mathbb{F}}}}
\newcommand{\G}{\ensuremath{{\mathbb{G}}}}
\newcommand{\HH}{\ensuremath{{\mathbb{H}}}}
\newcommand{\X}{\ensuremath{{\mathbb{X}}}}
\newcommand{\Y}{\ensuremath{{\mathbb{Y}}}}
\newcommand{\U}{\ensuremath{{\mathbb{U}}}}

\newcommand{\inner}[2]{{\langle #1, #2 \rangle}}
\newcommand{\verteq}{\rotatebox{90}{$\,=$}}
\newcommand{\Patt}[2]{\ensuremath{\mathcal{L}_{#2}(#1)}}
\newcommand{\Lang}[1]{\ensuremath{\mathcal{L}(#1)}}
\renewcommand{\vec}[1]{\mathbf{#1}}
\newcommand{\XP}{\ensuremath{{\cal X}}}

\newcommand{\id}{\ensuremath{\mathbf{1}}}
\newcommand{\kernel}[1]{\ensuremath{\mathrm{ker}(#1)}}
\newcommand{\setker}[1]{\ensuremath{\mathrm{cut}(#1)}}
\newcommand{\diff}[2]{\ensuremath{\mathrm{diff}(#1,#2)}}
\newcommand{\SP}[1]{\ensuremath{\mathrm{ST}(#1)}}
\newcommand{\trace}[2]{\ensuremath{\mathrm{Tr}_{#2}(#1)}}
\newcommand{\limit}[1]{\ensuremath{\Omega_{#1}}}
\newcommand{\eq}[1]{\ensuremath{\mathrm{Eq}(#1)}}

\newtheorem{theorem}{Theorem}
\newtheorem{lemma}{Lemma}
\newtheorem{corollary}{Corollary}
\newtheorem{claim}{Claim}
\newtheorem{question}{Question}

\theoremstyle{definition}
\newtheorem{exm}{Example}

\newenvironment{example}{%
  \begin{exm}
}{%
  \end{exm} \begin{center}\rule{2cm}{0.4pt}\end{center}
}

\begin{document}

\title{Effective Projections on Group Shifts to Decide Properties of Group Cellular Automata}

\author[1]{Pierre B\'{e}aur}
\author[2]{Jarkko Kari}

\affil[1]{Laboratoire Interdisciplinaire des Sciences du Num\'{e}rique, Universit\'{e} Paris-Saclay}
\affil[2]{Department of Mathematics and Statistics, University of Turku, Finland}

\date{}

\maketitle

\abstract{
\begin{center}
\begin{minipage}{\dimexpr\paperwidth-10cm}
\noindent Many decision problems concerning cellular automata are known to be decidable in the case of algebraic cellular automata, that is, when  the state set has an algebraic structure and the automaton acts as a morphism. The most studied cases include finite fields, finite commutative rings and finite commutative groups. In this paper, we provide methods to generalize these results to the broader case of group cellular automata, that is, the case where the state set is a finite (possibly non-commutative) finite group. The configuration space is not even necessarily the full shift but a subshift -- called a group shift --  that is a subgroup of the full shift on $\Z^d$, for any number $d$ of dimensions. We show, in particular, that injectivity, surjectivity, equicontinuity, sensitivity and nilpotency are decidable for group cellular automata, and non-transitivity is semi-decidable. Injectivity always implies surjectivity, and jointly periodic points are dense in the limit set. The Moore direction of the Garden-of-Eden theorem holds for all group cellular automata, while the Myhill direction fails in some cases. The proofs are based on effective projection operations on group shifts that are, in particular, applied on the set of valid space-time diagrams of group cellular automata. This allows one to effectively construct the traces and the limit sets of group cellular automata. A preliminary version of this work was presented at the conference Mathematical Foundations of Computer Science 2020.
 \end{minipage}
 \end{center}
}
\bigskip


\section{Introduction}

Algebraic group shifts and group cellular automata operate on configurations that are colorings of the infinite
grid $\Z^d$ by elements of a finite group $\G$, called the state set.
The set $\G^{\Z^d}$ of all configurations, called the full shift, inherits the group structure
as the infinite cartesian power of $\G$.
A subshift (a set of configurations avoiding a fixed set of forbidden finite patterns) is a group shift if it is also a subgroup of $\G^{\Z^d}$.
Group shifts are known to be of finite type, meaning that they can be defined by forbidding a finite number of patterns.
A cellular automaton is a dynamical system on a subshift, defined by a uniform local update rule of states. A cellular automaton on a group shift
is called a group cellular automaton if it is also a group homomorphism.

In this work we demonstrate that group shifts and group cellular automata in arbitrarily high dimensions $d$
are amenable to effective manipulations and
algorithmic decision procedures. This is in stark contrast to the general setup of
multidimensional subshifts of finite type and cellular automata where most properties are undecidable.
are plagued by undecidability.
Our considerations generalize a long line of past results -- see for example~\cite{DennunzioTCS2020,DennunzioPaper2019} and citations therein -- on algorithms for linear cellular automata (whose state set is a finite commutative ring)
and additive cellular automata (whose state set is a finite abelian group)  to non-commutative
groups and to arbitrary dimensions, and from the full shift to arbitrary group shifts.
Our methods are based on two classical results on group shifts: all group shifts -- in any dimension -- are of finite type,
and they have dense sets of periodic points~\cite{kitchens_schmidt_1989,schmidtBOOK}.
By a standard argumentation these provide a decision procedure for the membership in the
language of any group shift. We show how to use this procedure to effectively construct any lower dimensional projection
of a given group shift (Corollary~\ref{cor:projection}), and to construct the image of a given group shift under any given group cellular automaton
(Corollary~\ref{cor:homorphicimage}).

To establish decidability results for $d$-dimensional group cellular automata we then view the set of valid space-time diagrams as a $(d+1)$-dimensional
group shift. The local update rule of the cellular automaton provides a representation of this group shift.
The one-dimensional projections in the temporal direction are the trace subshifts of the automaton that provide all possible
temporal evolutions for a finite domain of cells, and the $d$-dimensional projection in the spatial dimensions is the limit set
of the automaton. These can be effectively constructed.
From the trace subshifts -- which are one-dimensional group shifts themselves --
one can analyze the dynamics of the cellular automaton and to decide, for example, whether it is periodic (Theorem~\ref{thm:decideperiodicity}),
equicontinuous or sensitive to initial conditions (Theorem~\ref{thm:sensitive}).
There is a dichotomy between equicontinuity and sensitivivity (Lemma~\ref{lem:dichotomy}).
We can semi-decide negative instances of mixing properties, i.e.,
non-transitive and non-mixing cellular automata (Theorem~\ref{thm:mixing}).
The limit set reveals whether the automaton is nilpotent (Theorem~\ref{thm:decideperiodicity}),
surjective or
injective (Theorem~\ref{thm:injectivity}).
Note that all these considerations work for group cellular automata over arbitrary group shifts, not only over
full shifts, and in all dimensions. We also note that in our setup injectivity implies surjectivity (Corollary~\ref{cor:injectiveimpliessurjective}) and that surjectivity implies pre-injectivity
(Theorem~\ref{thm:group_goe}), with neither implication holding in the inverse direction in general.
Moreover, in all surjective cases jointly spatially and temporally periodic points are dense (Corollary~\ref{cor:densetemporallyperiodic}).

The paper is structured as follows. We start by providing the necessary terminology and classical results about
shift spaces and cellular automata; first in the general context of multidimensional
symbolic dynamics and then in the algebraic setting in particular. In Section~\ref{sec:groupshifts}
we define projection operations on group shifts and exhibit effective algorithms to implement them.
This involves the main technical proof of the paper.
In Section~\ref{sec:groupCA} we apply the projections on space-time diagrams of cellular automata
to effectively construct their traces and limit sets. These are then used to provide decision algorithms for a number of
properties concerning group cellular automata.

We presented a preliminary version of this work at the conference Mathematical Foundations of Computer Science (MFCS 2020)~\cite{mfcs2020}.
The present article adds the main proof in Section~\ref{sec:groupshifts} of how the projections can be effectively constructed,
and a new part in Section~\ref{sec:groupCA} concerning the Garden-of-Eden theorem.


\section{Preliminaries}
\label{sec:prelminaries}

We first give definitions related to general subshifts and cellular automata, and then discuss concepts and properties
particular to group shifts and group cellular automata.

\subsection*{Symbolic dynamics}

A $d$-dimensional \emph{configuration}  over a finite alphabet $A$ is an assignment of symbols of $A$ on the infinite grid $\Z^d$. We call the elements of $A$
the \emph{states}.
For any configuration $c\in A^{\Z^d}$ and any cell $\vec{u}\in\Z^d$, we denote by $c_{\vec{u}}$ the state $c(\vec{u})$
that $c$ has in the cell $\vec{u}$. For any $a\in A$ we denote by $a^{\Z^d}$ the \emph{uniform} configuration defined by $a^{\Z^d}_{\vec{u}}=a$ for all $\vec{u}\in\Z^d$.

 For a vector $\vec{t}\in\Z^d$,  the \emph{translation} $\tau^{\vec{t}}$
shifts a configuration $c$ so that the cell $\vec{t}$ is pulled to the cell $\vec{0}$, that is,  $\tau^{\vec{t}}(c)_{\vec{u}}=c_{\vec{u}+\vec{t}}$ for all $\vec{u}\in\Z^d$.
We say that $c$ is \emph{periodic} if $\tau^{\vec{t}}(c)=c$ for some non-zero $\vec{t}\in\Z^d$. In this case $\vec{t}$ is a \emph{vector of periodicity} and $c$ is also termed
\emph{$\vec{t}$-periodic}. If there are $d$ linearly independent vectors of periodicity then $c$ is called \emph{totally periodic}.
We denote by $\vec{e}_i=(0,\dots,0,1,0\dots, 0)$ the basic $i$'th unit coordinate vector, for $i=1,\dots ,d$.
A totally periodic $c\in A^{\Z^d}$ has automatically, for some $k>0$, vectors of periodicity $k\vec{e}_1, k\vec{e}_2, \dots ,k\vec{e}_d$ in the $d$ coordinate directions.

Let $D\subseteq\Z^d$ be a finite set of cells, a \emph{shape}. A \emph{$D$-pattern} is an assignment $p\in A^D$  of symbols in the shape $D$. A
\emph{(finite) pattern} is a $D$-pattern for some shape $D$. We call $D$ the \emph{domain} of the pattern.
We say that a finite pattern $p$ of shape $D$ \emph{appears} in a configuration $c$ if for some $\vec{t}\in\Z^d$
we have $\tau^{\vec{t}}(c)|_{D}=p$. We also say that $c$ \emph{contains} the pattern $p$. For a fixed  $D$,
the set of $D$-patterns that appear in a configuration $c$ is denoted by $\Patt{c}{D}$. We denote by $\Lang{c}$ the set of all finite patterns that appear in $c$, i.e., the union of
$\Patt{c}{D}$ over all finite $D\subseteq\Z^d$.

Let $p\in A^D$ be a finite pattern of a shape $D$.
The set $[p]=\{c\in A^{\Z^d}\ |\ c|_{D}=p\}$ of configurations that have $p$ in the domain $D$ is called
the \emph{cylinder} determined by $p$.
The collection of cylinders $[p]$ is a base of a compact topology on $A^{\Z^d}$, the \emph{prodiscrete} topology.
See, for example, the first few pages of~\cite{tullio} for details.
The topology is equivalently defined by a metric on $A^{\Z^d}$ where two configurations are close to each other if they agree
with each other on a large region around the cell $\vec{0}$. Cylinders are clopen in the topology: they are both open and closed.

A subset $X$ of $A^{\Z^d}$ is called a \emph{subshift} if it is closed in the topology
and closed under translations.
Note that -- somewhat nonstandardly --  we allow $X$ to be the empty set.
By a compactness argument one has that every configuration $c$ that is not in $X$ contains a finite pattern $p$
that prevents it from being in $X$: no configuration that contains $p$ is in $X$. We can then as well  define subshifts  using forbidden patterns:
given a set $P$ of finite patterns we define
$$\XP_P=\{c\in A^{\Z^d}\ |\  \Lang{c}\cap P=\emptyset\},$$
the set of configurations that do not contain any of the patterns in $P$.
The set $\XP_P$ is a subshift, and every subshift is $\XP_P$ for some $P$.
If $X=\XP_P$ for some finite $P$ then $X$ is a \emph{subshift of finite type} (SFT).
For a subshift $X\subseteq A^{\Z^d}$ we denote by $\Patt{X}{D}$ and $\Lang{X}$ the sets of the $D$-patterns and all finite patterns that appear in
elements of $X$, respectively. The set $\Lang{X}$ is called the \emph{language} of the subshift.

A continuous function $F:X\longrightarrow Y$ between $d$-dimensional subshifts $X\subseteq A^{\Z^d}$ and $Y\subseteq B^{\Z^d}$ is a \emph{shift homomorphism}
if it is translation invariant, that is,  $\tau^{\vec{t}}_Y\circ F= F\circ \tau^{\vec{t}}_X$ for every $\vec{t}\in\Z^d$, where we have
denoted the translations $\tau^{\vec{t}}$ by a vector $\vec{t}$ with a subscript that indicates the space. A shift homomorphism from a subshift
$X$ to itself (i.e. a shift endomorphism) is called a \emph{cellular automaton} on $X$. The Curtis-Hedlund-Lyndon-theorem~\cite{hedlund}
states that shift homomorphisms are precisely the functions $X\longrightarrow Y$
defined by a local rule as follows. Let $N\subseteq\Z^d$ be a finite \emph{neighborhood}
and let $f:\Patt{X}{N}\longrightarrow B$ be a \emph{local rule} that assigns a letter of $B$ to every $N$-pattern that appears in $X$.
Applying $f$ at each cell yields a function $F_f:X\longrightarrow B^{\Z^d}$ that maps
every $c$ according to $F_f(c)_{\vec{u}}=f(\tau^{\vec{u}}(c)|_{N})$ for all $\vec{u}\in\Z^d$.
Shift homomorphisms $X\longrightarrow Y$ are precisely
such functions $F_f$ that also satisfy $F_f(X)\subseteq Y$.

The image $F(X)$ of a subshift under a shift homomorphism $F$ is clearly also a subshift. Images of subshifts of finite type
are called \emph{sofic}. We refer to~\cite{kurkaBOOK,LindMarcus95} for more concepts and results on symbolic dynamics.

\subsection*{Group shifts and group cellular automata}

Let $\G$ be a finite (not necessarily commutative) group.
There is a natural group structure on the $d$-dimensional
configuration space $\G^{\Z^d}$ where the group operation is applied cell-wise: $(ce)_{\vec{u}}=c_{\vec{u}}e_{\vec{u}}$ for all
$c,e\in\G^{\Z^d}$ and $\vec{u}\in \Z^d$. A \emph{group shift} is
a subshift of $\G^{\Z^d}$ that is also a subgroup. In particular, a group shift is not empty.
A cellular automaton $F:\X\longrightarrow \X$ on a group shift $\X\subseteq \G^{\Z^d}$
is a \emph{group cellular automaton} if it is a group homomorphism: $F(ce)=F(c)F(e)$ for all $c,e\in\X$.
More generally, a shift homomorphism $F:\X\longrightarrow \Y$ that is also a group homomorphism between groups shifts $\X$ and $\Y$
is called a \emph{group shift homomorphism}.

Group shifts have two
important properties that are central in algorithmic decidability~\cite{groupshift2}:
every group shift is of finite type, and totally
periodic configurations are dense in all group shifts~\cite{kitchens_schmidt_1989,schmidtBOOK}.

\begin{theorem}[\cite{kitchens_schmidt_1989}]
\label{thm:groupSFT}
Every group shift is a subshift of finite type.
\end{theorem}

It follows from this theorem that every group shift $\X$ has a finite representation using
a finite collection $P$
of forbidden finite patterns as $\X=\XP_P$. This is the representation assumed in all algorithmic questions concerning
given group shifts. Also when we say that we effectively construct a group shift $\X$
we mean that we produce a finite set $P$ of finite patterns such that $\X=\XP_P$.

\begin{theorem}[\cite{kitchens_schmidt_1989}]
\label{thm:periodic}
Totally periodic configurations are dense in group shifts, i.e., for every $p\in \Lang{\X}$ there is a totally periodic $c\in\X$ such that $p\in\Lang{c}$.
\end{theorem}

As an immediate corollary of these two fundamental properties we get that the language of a group shift is (uniformly) recursive.

\begin{corollary}
\label{cor:grouplanguage}
There is an algorithm that determines, for any
given group shift $\X\subseteq\G^{\Z^d}$ and any given finite pattern $p\in\G^D$
whether $p$ is in the language $\Lang{\X}$ of $\X$.
\end{corollary}

\begin{proof}
This is a standard argumentation by Hao Wang~\cite{wang}: There is a (non-deterministic) semi-algorithm for positive membership $p\in\Lang{\X}$
that guesses a totally periodic configuration $c\in \G^{\Z^d}$, verifies that $c$ contains the pattern $p$, and finally
verifies that $c$ does not contain any of the forbidden patterns in the given set $P$
that defines $\X=\XP_P$. Such a configuration $c$ exists by Theorem~\ref{thm:periodic}
iff $p\in\Lang{\X}$. Conversely, as for any SFT, there is a semi-algorithm for the negative cases $p\not\in\Lang{\X}$ that guesses a number $n$, makes sure that
the domain $D$ of $p\in \G^D$ is a subset of $E=\{-n,\dots,n\}^d$,
enumerates all finitely many patterns $q$ with domain $E$ that satisfy $q|_{D}=p$, and verifies that all such $q$ contain a copy of
a forbidden pattern in $P$ that defines $\X=\XP_P$. By compactness such a number $n$ exists iff $p\not\in\Lang{\X}$.
\end{proof}
The representation of an SFT in terms of forbidden patterns is not unique. However, as soon as the language is recursive,
we can effectively test if given representations define the same SFT.
\begin{corollary}
\label{cor:groupequivalnece}
There are algorithms to determine 
\begin{enumerate}
\item[(a)] whether $\X_1\subseteq \X_2$ holds for given group shifts $\X_1,\X_2\subseteq\G^{\Z^d}$,
\item[(b)] whether $\X_1 = \X_2$ holds for given group shifts $\X_1,\X_2\subseteq\G^{\Z^d}$,
\end{enumerate}
\end{corollary}
\begin{proof}
To prove (a), let $P=\{p_1,\dots ,p_k\}$ be the given set of forbidden
patterns that defines $\X_2=\XP_P$. We have $\X_1\subseteq \X_2$ if and only if $p_1,\dots ,p_k\not\in\Lang{\X_1}$, so (a)
follows from Corollary~\ref{cor:grouplanguage}. Now  (b) follows trivially from (a) and the fact that
$\X_1=\X_2$ iff $\X_1\subseteq \X_2$ and $\X_2\subseteq \X_1$.
\end{proof}
Another important known property is that there are no infinite strictly decreasing chains
$\X_1\supsetneq  \X_2\supsetneq  \X_3\supsetneq  \dots$ of group shifts~\cite{kitchens_schmidt_1989}.
This is clear as the intersection $\X$ of such a chain is a group shift and hence, by
Theorem~\ref{thm:groupSFT}, there is a finite set $P$ such that $\X=\XP_P$. If a pattern $p$ is in the
languages of all $\X_k$ in the chain then $p$ is also in the language of the intersection $\X$,
proving that for large enough $k$ the language of $\X_k$ does not contain any of the forbidden
patterns in $P$. This implies that $\X_k=\X$ and the chain does not decrease any further.
(Note, however, that while we presented here the decreasing chain property as a corollary to
Theorem~\ref{thm:groupSFT}, in reality the proof is interweaved in the proof of
Theorem~\ref{thm:groupSFT}, see~\cite{kitchens_schmidt_1989}.)

\begin{theorem}[\cite{kitchens_schmidt_1989}]
\label{thm:nodecreasingchain}
There does not exist an infinite chain $\X_1\supsetneq  \X_2\supsetneq  \X_3\supsetneq  \dots$
of group shifts $\X_i\subseteq \G^{\Z^d}$.
\end{theorem}

We also mention the obvious fact that pre-images of group shifts under group shift homomorphisms
$F:\X\longrightarrow \HH^{\Z^d}$ are group shifts and they can be
effectively constructed. In particular, this applies to the kernel $\kernel{F}=F^{-1} (\id^{\Z^d}_{\HH})$ of $F$.
(We denote the identity element of any group $\G$ by $\id_{\G}$, or simply by $\id$ if the group is clear from the context.)


\begin{lemma}
\label{lem:kernel}
For any given $d$-dimensional group shifts $\X\subseteq \G^{\Z^d}$ and $\Y\subseteq \HH^{\Z^d}$, and for a given
group shift homomorphism $F:\X\longrightarrow \HH^{\Z^d}$,
the set $F^{-1}(\Y)$ is a group shift that can be effectively constructed. In particular,
the kernel $\kernel{F}$ is a group shift that can be effectively constructed.
\end{lemma}

\begin{proof}
The set $F^{-1}(\Y)$ is clearly topologically closed, translation invariant, and a group, and therefore it is a group shift.
Let $P$ and $Q$ be the given finite sets of forbidden patterns defining $\X=\XP_P$ and $\Y=\XP_Q$.
Let $f:\Patt{\X}{N}\longrightarrow \HH$ be the given local rule with neighborhood $N\subseteq\Z^d$ that defines $F=F_f$. For each
forbidden $q\in\HH^D$ in $Q$  we forbid all patterns $p\in \G^{D+N}$ that the local rule maps to $q$. We also forbid all patterns $p\in P$. The resulting subshift of finite type is $F^{-1}(\Y)$.
\end{proof}

\section{Algorithms for group shifts}
\label{sec:groupshifts}

To effectively manipulate group shifts we need algorithms to perform some basic operations. The main operations we consider are
taking projections, either to lower the dimension of the space or to project into a subgroup of the state set
but keeping the dimension. As a byproduct we obtain
an algorithm to compute the image of a given group shift under a given group cellular automaton.
We use derivatives of the symbol $\pi$ for projections from $\Z^d$ to lower dimensional grids, and derivatives of the symbol
$\psi$  for projections that keep the dimension of $\Z^d$ but change the state set.

\subsection*{Notations for projections to lower dimensions}

Let us first define the projection operators that cut from $d$-dimensional configurations $(d-1)$-dimensional
slices of finite width in the first dimension. Let $d\geq 1$
be the dimension and $n\geq 1$ the width of the slice. For any $d$-dimensional configuration $c\in A^{\Z^d}$ over alphabet $A$
the $n$-\emph{slice} $\pi^{(n)}(c)$ is the $(d-1)$-dimensional configuration over alphabet $A^n$ that has in any cell $\vec{u}\in\Z^{d-1}$ the
$n$-tuple $(c(1,\vec{u}),\dots ,c(n,\vec{u}))\in A^n$. The $n$-slice of a subshift $X\subseteq A^{\Z^d}$
is then the set $\pi^{(n)}(X)$ of the $n$-slices of
all $c\in X$. Due to translation invariance of $X$, the fact that we cut slices at first coordinate positions $1,\dots ,n$ is irrelevant:
we could use any $n$ consecutive first coordinate positions instead.
Clearly  $\pi^{(n)}(X)$ is a subshift,
and if $\X\subseteq \G^{\Z^d}$ is a group shift then $\pi^{(n)}(\X)$ is also a group shift over the group
$\G^n=\G\times\dots \times \G$,
the $n$-fold cartesian power of $\G$. Note that the projection $\pi^{(n)}(X)$ of a subshift of finite type is not necessarily of finite type -- basically any effectively closed subshift can arise this way~\cite{Hochman} -- so group shifts are
particularly well behaving as their projections are of finite type.

Patterns in $(d-1)$-dimensional
slices of thickness $n$ can be interpreted in a natural way as $d$-dimensional patterns having
the width $n$ in the first dimension.
We introduce the notation $\hat{p}$ for such an interpretation of a pattern $p$. More precisely,
for any $D\subseteq \Z^{d-1}$ and a $(d-1)$-dimensional pattern $p\in (G^n)^D$ over the alphabet $\G^n$
we denote by $\hat{p}\in \G^E$ the corresponding $d$-dimensional
pattern over $\G$ whose domain is
$E=\{1,\dots ,n\}\times D\subseteq \Z^d$ and
$p(\vec{u}) = (\hat{p}(1,\vec{u}), \hat{p}(2,\vec{u}), \dots ,\hat{p}(n,\vec{u}))$
for every $\vec{u}\in D$. For a subshift $X$ we then have that $p\in\Lang{\pi^{(n)}(X)}$ if and only if $\hat{p}\in\Lang{X}$. In particular,
using an algorithm for the membership of a pattern in $\Lang{X}$ we can also decide the membership of any given finite pattern in $\Lang{\pi^{(n)}(X)}$.
Based on Corollary~\ref{cor:grouplanguage} we then have immediately the following fact for groups shifts.

\begin{lemma}
\label{lem:projectionlanguage}
One can effectively decide for any given $d$-dimensional group shift $\X\subseteq \G^{\Z^d}$, any given $n\geq 1$ and any given
$(d-1)$-dimensional finite pattern $p\in (G^n)^D$ whether $p\in\Lang{\pi^{(n)}(\X)}$.
\qed
\end{lemma}

Projections $\pi^{(n)}(X)$ are elementary slicing
operations that can be composed together, as well as with permutations of
coordinates, to obtain more general projections of subshifts into lower dimensional grids.
Very generally, for any subset
$E\subseteq\Z^d$ we call the restriction $c|_{E}$ the \emph{projection} of $c$ on $E$, and
the projection of a subshift $X$ on $E$ is $\pi_E(X)=\{c|_{E}\ |\ c\in X\}$.
We mostly use operation $\pi_E$ with sets of type
$E=D\times \Z^k$ for some $k<d$ and a finite $D\subseteq \Z^{d-k}$, and
we mostly apply $\pi_E$ to group shifts $\X\subseteq \G^{\Z^d}$.
The projection $\pi_E(\X)$ is then viewed in the natural manner
as the $k$-dimensional group shift over the finite group $\G^{D}$. One of the main results of this section is
Corollary~\ref{cor:projection}, stating that we can effectively construct $\pi_E(\X)$ for given $\X$ and $E=D\times \Z^k$.


\subsection*{Notations for projections that keep the dimension}

Let $\G=\G_1\times \G_2$ be a cartesian product of two finite groups. For any $c\in \G^{\Z^d}$ we let
$\psi^{(1)}(c)\in \G_1^{\Z^d}$ and $\psi^{(2)}(c)\in \G_2^{\Z^d}$
be the cell-wise projections to $\G_1$ and $\G_2$, respectively, defined by $c_{\vec{u}}=(\psi^{(1)}(c)_{\vec{u}}, \psi^{(2)}(c)_{\vec{u}})$ for all $\vec{u}\in\Z^d$. By abuse of notation, for any  $c^{(1)}\in \G_1^{\Z^d}$
and $c^{(2)}\in \G_2^{\Z^d}$ we denote by $(c^{(1)}, c^{(2)})$ the configuration $c\in (\G_1\times \G_2)^{\Z^d}$
such that $\psi^{(i)}(c)=c^{(i)}$ for $i=1,2$. We also use the similar notation on finite patterns
and implicitly use the obvious way to identify $\G_1^{D}\times \G_2^{D}$ and $(\G_1\times \G_2)^{D}$.

Clearly, for any group shift $\X\subseteq\G^{\Z^d}$, the sets
$\psi^{(1)}(\X)$ and $\psi^{(2)}(\X)$ are group shifts over $\G_1$ and $\G_2$, respectively.
A pattern $p\in (\G_1)^D$ is in the language of $\psi^{(1)}(\X)$ if and only if there is a pattern
$q\in (\G_2)^D$ such that $(p,q)\in\Patt{\X}{D}$.
Therefore we have the following counter part of Lemma~\ref{lem:projectionlanguage}.

\begin{lemma}
\label{lem:restrictionlanguage}
One can effectively decide for any given $d$-dimensional group shift $\X\subseteq (\G_1\times \G_2)^{\Z^d}$,
and any given $d$-dimensional finite pattern $p\in (\G_1)^D$ whether $p\in\Lang{\psi^{(1)}(\X)}$.
\qed
\end{lemma}

Let $D,E$ be finite sets, $D\subseteq E$, and let $\X\subseteq (\G^E)^{\Z^d}$ be a group shift over the finite cartesian power
$\G^E$ of the group $\G$. The group $\G^E$ is isomorphic to $\G^D\times\G^{E\setminus D}$ in a natural manner,
and $\psi^{(1)}$ projects then $\X$
into $(\G^D)^{\Z^d}$. We denote this projection by $\psi_D$. Notice that $\pi_{D\times\Z^k}=\psi_{D}\circ\pi_{E\times\Z^k}$
so that the projection into $D\times\Z^k$ can be obtained as a composition of projections $\pi^{(n)}$ into slices, permutations
of coordinates, and a projection of the type $\psi^{(1)}$.

\subsection*{Effective constructions}

Our main technical result is that projections of group shifts can be effectively constructed. We state this as a two-part lemma.
Corollaries~\ref{cor:projection} and~\ref{cor:homorphicimage} that follow the lemma
provide clean statements that we use in the rest of the paper.

\begin{lemma}
\label{lem:main}
Let $d\geq 1$ be a dimension, and let $\G$ and $\G_1,\G_2$ be finite groups.
\begin{enumerate}
\item[(a)] For any given $d$-dimensional group shift $\X\subseteq \G^{\Z^d}$ and any given $n\geq 1$ one can effectively construct the
$d-1$ dimensional group shift $\pi^{(n)}(\X)\subseteq (\G^n)^{\Z^{d-1}}$.
\item[(b)] For any given $d$-dimensional group shift $\X\subseteq (\G_1\times \G_2)^{\Z^d}$
one can effectively construct the $d$-dimensional group shift $\psi^{(1)}(\X)\subseteq \G_1^{\Z^d}$.
\end{enumerate}
\end{lemma}

\begin{proof}
The proof is by induction on dimension $d$. We first prove (a) for dimension $d$ assuming that (b) holds in dimension $d-1$, and then we prove (b) for dimension $d$ assuming
(a) holds in dimension $d$ and that (b) holds for dimension $d-1$.
To start the induction we observe that (b) trivially holds for dimension $d=0$: In this case group shifts over $\G$ are precisely subgroups of $\G$.
\bigskip

\noindent
\emph{Proving (a) for dimension $d$ assuming (b) holds for dimension $d-1$:}
Let a width $n\geq 1$ and a group shift $\X\subseteq \G^{\Z^d}$ be given
(in terms of a finite set $P$ of forbidden patterns such that $\X=\XP_P$). Let us first assume that $n$ is at least the width of the patterns in $P$
so that we can assume that all patterns in $P$ have the same domain
$\{1,\dots,n\}\times D$ for some finite $D\subseteq\Z^{d-1}$. (Note that we can effectively grow the domain of each forbidden pattern by forbidding
instead all patterns with the larger domain that extend the original pattern. Thus a common domain can be taken for all elements in $P$.
We can also shift the domains of the patterns.)

To construct the $(d-1)$-dimensional projection $\Y=\pi^{(n)}(\X)$
we effectively enumerate and forbid patterns that are not in the language of $\Y$. We accumulate
the forbidden patterns in a set $Q$ that we initialize to be the empty set in the beginning of the process.
Let $D_1, D_2,\dots$ be an effective enumeration of all finite subsets of $\Z^{d-1}$ with $D_1=D$.
For each $i=1,2,\dots$ in turn we go through all (finitely many) patterns $q$ over $\G^n$ having shape $D_i$
and check, using Lemma~\ref{lem:projectionlanguage}, whether  $q$ is in $\Lang{\Y}$. If not, we
add $q$ in the set $Q$. This way, at any time, $Q$ only contains patterns outside of $\Lang{\Y}$ and hence forbidding patterns in $Q$ gives an
upper approximation $\XP_Q\supseteq \Y$. Since  $\Y$ is a group shift and therefore of finite type, by
systematically enumerating the patterns in the complement of $\Lang{\Y}$ we  eventually reach a set $Q$ such that $\Y=\XP_Q$.

The reason why we process all patterns for each shape $D_i$ before moving to the next shape $D_{i+1}$
is the observation that this way the subshift $\XP_Q$ is guaranteed to be a group shift after finishing processing
$D_i$. We have the following general fact:

\begin{claim}
\label{claim2}
Let $\X\subseteq\G^{\Z^d}$ be any group shift in any dimension $d$, and let $D\subseteq\Z^d$ be finite. For $Q=\G^D\setminus\Patt{\X}{D}$
the subshift $\XP_Q$ is a group shift and $\X\subseteq\XP_Q$.
\end{claim}

\begin{proof}[Proof of Claim~\ref{claim2}]
Clearly $\XP_Q$ is a subshift and $\X\subseteq\XP_Q$. We just have to show that $\XP_Q$ is a group.
We have $c\in\XP_Q$ if and only if $\Patt{c}{D}\subseteq\Patt{\X}{D}$. The result now follows from the fact that
$\Patt{\X}{D}$ is a subgroup of $\G^D$.
\end{proof}

Intersections of group shifts are group shifts so Claim~\ref{claim2} indeed implies that after fully processing
any number of domains $D_1,\dots ,D_i$ the resulting subshift $\XP_Q\supseteq \Y$ is a group shift. Note also that
$D_1=D$ guarantees that already after the first round $i=1$ we have in $Q$ all the patterns of $P$.

As mentioned above, we are guaranteed to eventually have enough forbidden patterns in $Q$ to have $\Y=\XP_Q$.
The problem is to identify when we have enumerated enough patterns and reached such a set $Q$. Fortunately this can be detected by
checking that the left and the right slices of width $n-1$ of the upper approximation $\XP_Q$ are identical with each other,
as detailed below.

Let us introduce notations $\psi_L$ and $\psi_R$ for the operations of
extracting the left and the right slices of width $n-1$. More precisely,
for a configuration $c=(c^{(1)},\dots ,c^{(n)})\in (\G^n)^{\Z^{d-1}}$ of thickness $n$,
where $c^{(i)}\in \G^{\Z^{d-1}}$ are the single cell wide slices of $c$,
we define
$\psi_L(c)=(c^{(1)},\dots ,c^{(n-1)})$ and $\psi_R(c)=(c^{(2)},\dots ,c^{(n)})$, respectively. Both are elements of $(\G^{n-1})^{\Z^{d-1}}$.

\begin{claim}
\label{claim1}
$\XP_Q=\Y$ if and only if $\psi_L(\XP_Q)=\psi_R(\XP_Q)$.
\end{claim}

\begin{proof}[Proof of Claim~\ref{claim1}]
$\psi_L(\Y)=\psi_R(\Y)=\pi^{(n-1)}(\X)$ so the implication from left to right is clear. For the converse direction,
let $c\in \XP_Q$ be arbitrary. By the assumption $\psi_L(\XP_Q)=\psi_R(\XP_Q)$ there exists a bi-infinite sequence $\dots ,c_{-1},c_0, c_1,\dots$
of configurations such that $c_0=c$ and for all $i\in\Z$ we have $c_i\in \XP_Q$ and $\psi_R(c_i)=\psi_L(c_{i+1})$.
Configurations $c_i$ and $c_{i+1}$ overlap properly so that there is a $d$-dimensional configuration $c'\in \G^{\Z^d}$
whose consecutive $n$-slices are $\dots ,c_{-1},c_0, c_1,\dots$, that is, $c_i = \pi^{(n)}(\tau_{i\vec{e}_1}(c'))$ for all
$i\in\Z$. See Figure~\ref{fig:overlaps} for an illustration of $c'$.
Since each forbidden pattern in $P$ is also in $Q$, none of the slices contain such a forbidden pattern and hence $c'\in \X$.
Now $c=c_0=\pi^{(n)}(c')$ so that $c\in \pi^{(n)}(\X)=\Y$. We have shown that $\XP_Q\subseteq \Y$. The opposite inclusion holds since $\XP_Q$ is an
upper approximation of $\Y$.
\end{proof}

\begin{figure}[ht]
\begin{center}
\includegraphics[width=0.8\textwidth]{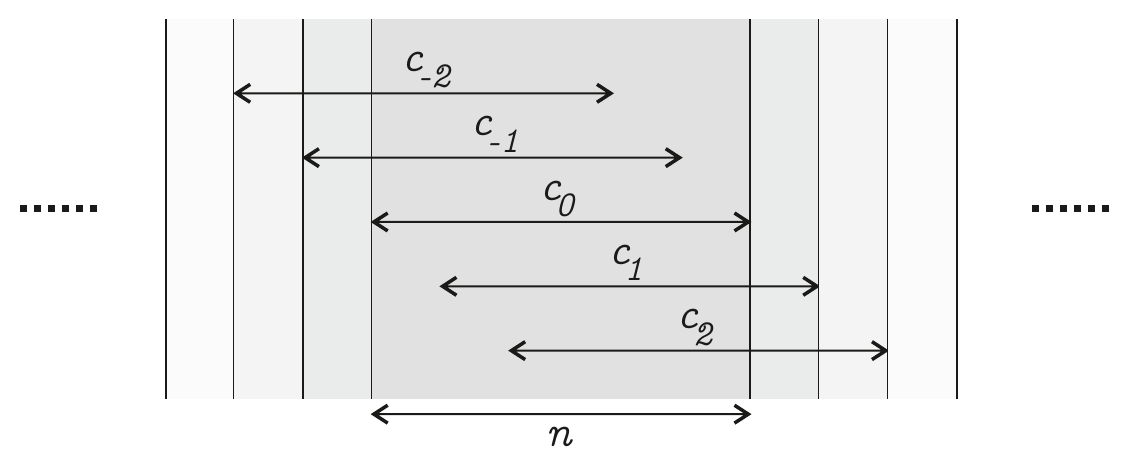}
\end{center}
\caption{An illustration of the overlapping $n$-slices forming the configuration $c'$ in the proof of Claim~\ref{claim1}.}
\label{fig:overlaps}
\end{figure}

Both $\psi_L$ and $\psi_R$ are projection operations of type (b) of the present lemma,
so by the inductive hypotheses and the fact that $\XP_Q$ is a $(d-1)$-dimensional group shift, the
group shifts $\psi_L(\XP_Q)$ and $\psi_R(\XP_Q)$ can be effectively constructed. Moreover, equality
of group shifts is decidable so that the condition $\psi_L(\XP_Q)=\psi_R(\XP_Q)$ can be effectively checked.
In conclusion, each time our algorithm finishes with adding patterns of shape $D_i$ in $Q$ it checks whether
$\psi_L(\XP_Q)=\psi_R(\XP_Q)$ holds for the current group shift $\XP_Q$. The algorithm stops and returns set $Q$
once equality is reached.
This finishes the description of the algorithm for case (a), provided $n$ is large enough to have all patterns $P$ in a
slice of width $n$. If $n$ is smaller, we first execute the algorithm for large enough width $m>n$ and effectively compute the further projection
$\pi^{(n)}(\X)=\psi_L^{m-n}(\pi^{(m)}(\X))$ to slices of width $n$. The projection can be effectively computed by the inductive hypothesis because it is a $(d-1)$-dimensional operation of type (b) of the present lemma.
\bigskip

\noindent
\emph{Proving (b) for dimension $d$ assuming that (a) holds for dimension $d$ and that (b) holds for dimension $d-1$:}
Let $\X\subseteq (\G_1\times\G_2)^{\Z^d}$ be given (in terms of a finite set $P$ of forbidden patterns such that $\X=\XP_P$).
To construct the
$d$-dimensional projection $\Y=\psi^{(1)}(\X)$ we -- analogously to the proof of case (a) above --
use Lemma~\ref{lem:restrictionlanguage} to
effectively enumerate patterns that are not in the language of $\Y$, thus obtaining upper approximations
of $\Y$ by subshifts $\XP_Q$. We process all patterns of a shape $D_i$ before moving on to the next shape $D_{i+1}$.
This guarantees -- as proved in Claim~\ref{claim2} above  -- that after finishing with each shape $D_i$ the shift $\XP_Q$
is a group shift.

We eventually reach a set $Q$ such that $\Y=\XP_Q$, but the challenge is again to identify
when we have reached such $Q$.
We establish this by proving that we can
effectively compute a number $n$ such that $\Y=\XP_Q$ if and only if  $\pi^{(n)}(\XP_Q)=\pi^{(n)}(\Y)$.

Once number $n$ is known, the projection $\pi^{(n)}(\XP_Q)$ can be effectively constructed by the inductive hypothesis stating that
case (a) of the present lemma holds in dimension $d$. Indeed, $\XP_Q$ is a known $d$-dimensional group shift.
The projection $\pi^{(n)}(\Y)$ can also be effectively constructed
as projections $\psi^{(1)}$ and $\pi^{(n)}$ commute, so that we can first construct $\pi^{(n)}(\X)$ (using the inductive
hypothesis that case (a) of the present lemma holds in dimension $d$) and then we apply $\psi^{(1)}$ on the $(d-1)$-dimensional
group shift $\pi^{(n)}(\X)$ (using the inductive
hypothesis that case (b) of the present lemma holds in dimension $d-1$) to obtain $\pi^{(n)}(\Y)$.

All that remains is to compute a sufficiently large $n$ for the implication
$$\pi^{(n)}(\XP_Q)=\pi^{(n)}(\Y) \Longrightarrow \XP_Q=\Y$$
to hold.

First a note on notations: Recall that we denote for any $c\in\G_1^{\Z^d}$ and $e\in\G_2^{\Z^d}$
by $(c,e)$ the configuration in $(\G_1\times \G_2)^{\Z^d}$ such that $\psi^{(1)}(c,e)=c$ and
$\psi^{(2)}(c,e)=e$. We also then denote for any $c\in (\G^n)^{\Z^d}$ and $c'\in (\G^m)^{\Z^d}$
by $(c,c')$ the concatenated configuration in $(\G^{n+m})^{\Z^d}$, by the understanding that $\G^{n+m}=\G^n\times\G^m$.
In the following we are going to mix both types of concatenations. For example, for
$c\in(\G_1^n)^{\Z^d}$, $c'\in (\G_1^m)^{\Z^d}$, $e\in(\G_2^n)^{\Z^d}$ and $e'\in(\G_2^m)^{\Z^d}$
we may write $((c,e),(c',e'))$ for a concatenated configuration in  $((\G_1\times \G_2)^{n+m})^{\Z^d}$, but also
$((c,c'),(e,e'))$ for the same configuration, now expressed in $(\G_1^{n+m}\times \G_2^{n+m})^{\Z^d}$. To help the reader
in the task of parsing such expressions, we use the notation $[c|c']$ for the second type of concatenations, with the idea that $n$-slices can be visualized as strips in the vertical direction and the vertical line $|$ is a ``separator'' between concatenated
vertical strips. So the two examples above will be written as
$[(c,e)|(c',e')]$ and $([c|c'],[e|e'])$, respectively.

Also note that for configurations $c$ and $e$ of the same group shift, say $c,e\in (\G^n)^{\Z^d}$, the notation $ce$
is not for the concatenation of the strips but it is for the cell-wise product of the configurations, \emph{i.e.}, for the product in the
group $(\G^n)^{\Z^d}$.

For any group shift $\U$ over
the alphabet $\G_1\times \G_2$ we denote by $\setker{\U}$ the set of configurations $c$ over $\G_2$
such that $(\id,c)\in \U$. Because $\setker{\U}=\psi^{(2)}(\kernel{\psi^{(1)}}\cap \U)$ and
because projections $\psi^{(i)}$ are group shift homomorphisms the set
$\setker{\U}$ is a group shift.

\begin{claim}
\label{claim3}
For any given group shift $\U\subseteq (\G_1\times \G_2)^{\Z^k}$, in any dimension $k$, one can effectively
construct $\setker{\U}$.
\end{claim}

\begin{proof}[Proof of Claim~\ref{claim3}]
By Lemma~\ref{lem:kernel} we can effectively construct $\kernel{\psi^{(1)}}$. Intersections of subshifts of finite type can
be effectively constructed (simply take the union of the defining sets of forbidden patterns of the two SFTs).
This means that $\U'=\kernel{\psi^{(1)}}\cap \U$ can be effectively constructed. Let $R$ be the constructed
set of finite patterns such that $\U'=\XP_R$.
All configurations in $\U'$ have $\id$ in their first components so
to define $\psi^{(2)}(\U')$ it is enough to forbid for all $(\id,p)\in R$ the pattern $p$.
\end{proof}


After these notations we can proceed with the proof. Let $m$ be a number such that the
forbidden patterns in set $P$ that defines $\X$ fit in a slice of thickness $m$, that is,
the domain of each forbidden pattern in $P$ is a subset of $\{1,\dots ,m\}\times\Z^{d-1}$.
Let us call a positive integer $r$ a \emph{radius of synchronization} if
for all $w\in \G_2^{\{1,\dots ,m\}\times\Z^{d-1}}$ holds the implication
\begin{equation}
\label{eq1}
\begin{array}{l}
(\exists u,v\in \G_2^{\{1,\dots ,r\}\times\Z^{d-1}})\ [u|w|v] \in \setker{\pi^{(m+2r)}(\X)}\hspace*{5mm}
\vspace*{3mm}
\\
\hspace*{50mm} \Longrightarrow \hspace*{5mm} w\in\pi^{(m)}(\setker{\X}).
\end{array}
\end{equation}
(See Figure~\ref{fig1} for an illustration.)

\begin{figure}[ht]
\begin{center}
\includegraphics[width=0.9\textwidth]{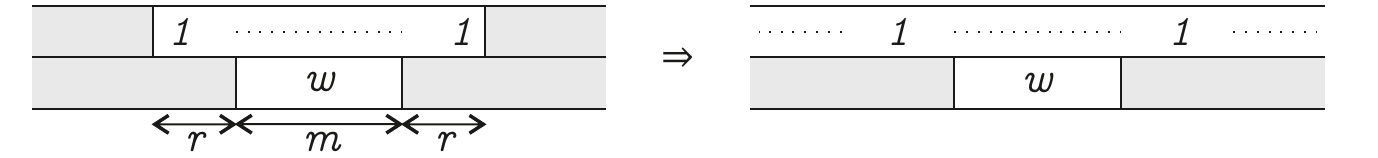}
\end{center}
\caption{An illustration of the implication (\ref{eq1}). Upper and lower layers are configurations over $\G_1$ and $\G_2$, respectively. Letter $r$ is a radius of  synchronization if for every $w$ for which the left situation exists in $\X$ also the right situation exists in $\X$. The picture depicts only the first dimension -- each letter represents an entire $(d-1)$-dimensional configuration.
}
\label{fig1}
\end{figure}

\begin{claim}
\label{claim4}
A radius of synchronization exists, and we can effectively find one.
\end{claim}

\begin{proof}[Proof of Claim~\ref{claim4}]
For any $r$, let us denote by $\U_r$ the set of $w\in \G_2^{\{1,\dots ,m\}\times\Z^{d-1}}$ that satisfy the left-hand-side of
implication of (\ref{eq1}), and by $\U$ the set of those that satisfy the right-hand-side. Now $\U=\pi^{(m)}(\setker{\X})$ and
$\U_r=\psi(\setker{\pi^{(m+2r)}(\X)})$
where $\psi$ is the projection in the central segment of length
$m$. It follows that $\U_r$ and $\U$
are $d-1$-dimensional group shifts. Group shifts $\U_r$ form a decreasing chain $\U_1\supseteq \U_2\supseteq\dots$ so by
Theorem~\ref{thm:nodecreasingchain} there exists $k$ such that $\U_r=\U_k$ for all $r\geq k$. By a simple compactness argument we then also have that $\U=\U_k$: if $w\in\U_k$ then for every $r\geq k$ there exists $c_r\in\X$ as in the left of Figure~\ref{fig1},
so that a limit of a converging subsequence of $c_k,c_{k+1},\dots$ is as in the right of Figure~\ref{fig1}, proving that $w\in\U$.
This proves that $k$ is a radius of synchronization.

To find a radius of synchronization we enumerate $r=1,2,\dots$ and test for each $r$ whether $\U_r=\U$. This can be effectively
tested: First, by Claim~\ref{claim3} the set $\setker{\X}$ can be constructed and then
by the inductive hypothesis that (a) holds in dimension $d$ we can apply $\pi^{(m)}$ to form $\U$. Second,
by the inductive hypothesis that (a) holds in dimension $d$ we can construct $\pi^{(m+2r)}(\X)$, by Claim~\ref{claim3}
we can build $\setker{\pi^{(m+2r)}(\X)}$, and finally by the inductive hypothesis that (b) holds in dimension $d-1$
we apply $\psi$ to construct $\U_r$. So both $\U$ and $\U_r$ can be effectively constructed, and by Corollary~\ref{cor:groupequivalnece}(b) we can test whether they are equal.
\end{proof}

The importance of the radius of synchronization comes from the fact that sufficiently wide slices of identities $\id$ can be extended.
\begin{claim}
\label{claimX}
Let $r$ be a radius of synchronization. Then for any slice $x\in\pi^{(k)}(\Y)$ of any width $k$ holds the implication
$$
[x|\id^{m+2r}]\in\pi^{(k+m+2r)}(\Y)\hspace*{5mm} \Longrightarrow \hspace*{5mm} [x|\id^{m+2r+1}]\in\pi^{(k+m+2r+1)}(\Y).
$$
\end{claim}
\begin{proof}[Proof of Claim~\ref{claimX}]
Assume the left-hand-side of the implication.
Recalling that $\Y=\psi^{(1)}(\X)$ there is a configuration $c\in\X$ such that $\pi^{(k+m+2r)}(c) = [(x,y)| (\id^r,u)|(\id^m,w)|(\id^r,v)]$
for some slices $y,u,w,v$ (of thicknesses $k$, $r$, $m$ and $r$, respectively) over $\G_2$.
In particular then $[u|w|v]\in \setker{\pi^{(m+2r)}(\X)}$,
so that the implication (\ref{eq1}) gives that $w\in\pi^{(m)}(\setker{\X})$.
By the definition of $\setker{\X}$
there is a configuration $e\in \X$ such that $\pi^{(k+m+2r+1)}(e)=[(\id^k,y')| (\id^r,u')|(\id^m,w)|(\id^r,v')|(\id,a)]$
for slices $y',u',v'$ and $a$ of thicknesses $k$, $r$,  $m$, $r$ and $1$, respectively.
The forbidden patterns in the set $P$ that defines $\X$ have thickness at most $m$, so we can cut and exchange
tails at the common slice $(\id^m,w)$ of $c$ and $e$ without introducing any forbidden patterns. See Figure~\ref{fig:cutandpaste}
for an illustration of the cut and exchange between $c$ and $e$ along their common slice.
This implies that the slice $[(x,y)| (\id^r,u)|(\id^m,w)|(\id^r,v')|(\id,a)]$ of thickness $k+m+2r+1$
is in $\pi^{(k+m+2r+1)}(\X)$, providing the result that
$[x|\id^{m+2r+1}]\in\pi^{(k+m+2r+1)}(\Y)$.

\end{proof}

\begin{figure}[ht]
\begin{center}
\includegraphics[width=0.9\textwidth]{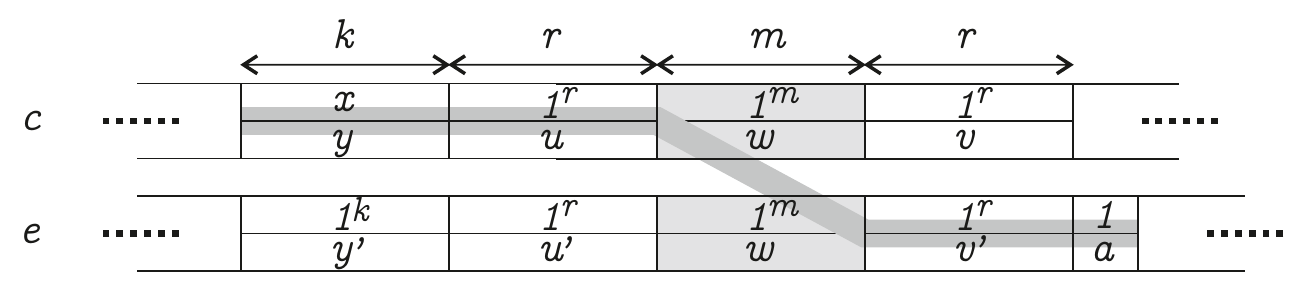}
\end{center}
\caption{An illustration of cutting and reconnecting halves of configurations $c$ and $e$ along a common slice of width $m$ in the proof of Claim~\ref{claimX}.}
\label{fig:cutandpaste}
\end{figure}

Let $n=m+2r+1$ where $r$ is the radius of synchronization that we computed for $\X$.
This turns out to be a sufficient thickness for our purpose of halting the algorithm.

\begin{claim}
\label{claimY}
If $\pi^{(n)}(\XP_Q)=\pi^{(n)}(\Y)$ then $\pi^{(k)}(\XP_Q)=\pi^{(k)}(\Y)$ for all $k\geq n$.
\end{claim}

\begin{proof}[Proof of Claim~\ref{claimY}]
We prove this by induction on $k$. Case $k=n$ is clear. For the inductive step, suppose that
$\pi^{(k)}(\XP_Q)=\pi^{(k)}(\Y)$ is known for some $k\geq n$ and consider slices of width $k+1$.
Containment $\pi^{(k+1)}(\Y)\subseteq \pi^{(k+1)}(\XP_Q)$
is clear since $\Y\subseteq\XP_Q$. We just need to prove that $\pi^{(k+1)}(\XP_Q)\subseteq \pi^{(k+1)}(\Y)$.

Let $c\in \pi^{(k+1)}(\XP_Q)$ so that $c=[a|x|b]$ for $a,b\in \pi^{(1)}(\XP_Q)$ and $x\in\pi^{(k-1)}(\XP_Q)$.
We have $[a|x], [x|b]\in \pi^{(k)}(\XP_Q)$ so that by the inductive hypothesis $[a|x], [x|b]\in\pi^{(k)}(\Y)$.
Because $[x|b]$ is a slice in a configuration of $\Y$, there exists $a'\in\pi^{(1)}(\Y)$ such that
$[a'|x|b]\in \pi^{(k+1)}(\Y)$. Because $\Y$ is a group shift the product $[a|x]\, [a'|x]^{-1}=[aa'^{-1}|\id^{k-1}]$
is in $\pi^{(k)}(\Y)$. Because $k-1\geq m+2r$ we get from Claim~\ref{claimX} that $[aa'^{-1}|\id^{k}]$
is in $\pi^{(k+1)}(\Y)$. But this is all we need: we get
$[aa'^{-1}|\id^{k}]\, [a'|x|b] = [a|x|b]=c$ is in $\pi^{(k+1)}(\Y)$ as claimed.

\end{proof}

It is now a simple compactness argument to show that if $\pi^{(k)}(\XP_Q)=\pi^{(k)}(\Y)$ for all $k\geq n$ then $\XP_Q=\Y$.
So our algorithm constructs sets $Q$ until condition  $\pi^{(n)}(\XP_Q)=\pi^{(n)}(\Y)$  is satisfied for $n=m+2r+1$.
At that time we can stop  because we know that we have reached the situation $\XP_Q=\Y$. This completes the proof of Lemma~\ref{lem:main}.\qed
\end{proof}

\bigskip

Lemma~\ref{lem:main} is used in the rest of the paper via the following two corollaries.
The first corollary states that arbitrary projections can be effectively implemented on group shifts.

\begin{corollary}
\label{cor:projection}
Given a $d$-dimensional group shift $\X\subseteq \G^{\Z^d}$ and given $k<d$ and a finite $D\subseteq \Z^{d-k}$
we can effectively construct the $k$-dimensional group shift $\pi_{D\times\Z^{k}}(\X)\subseteq (\G^D)^{\Z^{k}}$.
\end{corollary}

\begin{proof}
By shift invariance of $\X$ we arbitrarily translate $D$, so we may assume without loss of generality that
$D$ is a subset of $E=\{1,\dots ,n\}^{d-k}$ for some $n$.
By applying $d-k$ times Lemma~\ref{lem:main}(a), permuting the coordinates as needed, we can effectively construct
$\X'=\pi_{E\times \Z^{k}}(\X)$. Now $\pi_{D\times\Z^{k}}(\X)=\psi_{D}(\X')$, and by Lemma~\ref{lem:main}(b) the
projection $\psi_{D}$ from $\G^E$ to $\G^D$ can be effectively implemented.
\end{proof}

The second corollary tells that images of group shifts under group cellular automata can be also effectively constructed.

\begin{corollary}
\label{cor:homorphicimage}
Given a $d$-dimensional group shift $\X\subseteq \G^{\Z^d}$ and given a group shift homomorphism $F:\X\longrightarrow \HH^{\Z^d}$
one can effectively construct the group shift $F(\X)\subseteq \HH^{\Z^d}$.
\end{corollary}
\begin{proof}
Let $\X=\XP_P$ where $P$ is the given finite set of forbidden patterns that defines $\X$, and let $F=F_f$ where
$f:\Patt{\X}{N}\longrightarrow \HH$ is the given local rule of $F$ with a neighborhood $N$. We can pad symbols to patterns to
grow their domains, so we can assume without loss of generality that all patterns in $P$ have the same domain $D$,
that the neighborhood is the same set $N=D$, and that $\vec{0}\in D$.

We first effectively construct $\X' = \{(c,F(c))\ |\ c\in\X\}\subseteq (\G\times\HH)^{\Z^d}$.
This is a group shift over group $\G\times\HH$ because $F$ is a homomorphism. It is defined by
forbidding all patterns $(p,q)\in (\G\times\HH)^D$ where $p\not\in\Patt{\X}{D}$, or $p\in\Patt{\X}{D}$ but $q(\vec{0})\neq f(p)$.
So $\X'$ can indeed be effectively constructed.
By Lemma~\ref{lem:main}(b) we can then effectively compute the second projection $F(\X)=\psi^{(2)}(\X')$.
\end{proof}

\section{Algorithms for group cellular automata}
\label{sec:groupCA}

In this part we apply the algorithms developed for group shifts to analyze group cellular automata. The basic idea is to view the set of space-time diagrams as  a higher dimensional group shift and to effectively compute one-dimensional projections in the temporal direction. This way, trace subshifts are obtained. As these are one-dimensional group shifts, and hence of finite type,
the long term dynamics can be analyzed. A projection in the spatial dimensions provides the limit set of the cellular automaton.

We first define the central concepts of space-time diagrams, traces and limit sets, and show that they can be effectively constructed. Then we use this to prove properties and algorithms concerning several dynamical properties of group cellular automata.
We refer to~\cite{casurvey,kurkaBOOK} for more details and known results on the dynamical properties we consider.

\subsection*{Space-time diagrams}

Let $\X\subseteq\G^{\Z^d}$ be a $d$-dimensional group shift and let $F:\X\longrightarrow \X$ be a group cellular automaton on $\X$.
A bi-infinite \emph{orbit} of $F$ is a sequence $\dots c^{(-1)},c^{(0)},c^{(1)},\dots$ of configurations
$c^{(i)}\in\X$ such that $c^{(i+1)}=F(c^{(i)})$
for all $i\in\Z$. Such an orbit can be viewed as the $(d+1)$-dimensional configuration $c\in\G^{\Z^{d+1}}$ by concatenating the
configurations $c_i$ one after the other along the additional dimension, that is,
$c_{\vec{u},i}=c^{(i)}_{\vec{u}}$ for all $i\in\Z$ and $\vec{u}\in \Z^d$. The first $d$ dimensions are spatial dimensions while the
$(d+1)$st dimension is the temporal dimension. The configuration $c$ is a \emph{space-time diagram} of
the cellular automaton $F$. Note that the orbits and space-time-diagrams are temporally bi-infinite.
The set of all space-time diagrams of $F$ is denoted by $\SP{F}$. Because $F$ is
a group homomorphism we have the following.

\begin{lemma}
\label{lem:SPisgroupshift}
$\SP{F}\subseteq\G^{\Z^{d+1}}$ is a group shift.
\qed
\end{lemma}

Given $\X$ and $F$ we can effectively construct $\SP{F}$. Indeed, we just need to forbid in spatial slices
all the forbidden patterns that define $\X$, and in temporally consecutive pairs of slices patterns where
the local update rule of $F$ is violated. More precisely, let $P$ be the given finite
set of forbidden patterns that defines $\X=\XP_P$, and let $f:\Patt{\X}{N}\longrightarrow \G$
be the given local update rule that defines $F$ with the finite neighborhood $N\subseteq\Z^d$. For any $p\in P$
we forbid the $(d+1)$-dimensional pattern $\hat{p}$ over the domain $D\times\{0\}$ with $\hat{p}(\vec{u},0) = p(\vec{u})$
for all $\vec{u}\in D$, i.e., the spatial slices are forced to belong to $\X$, and for any neighborhood pattern
$q\in \Patt{\X}{N}$ and for any $a\in\G$ such that $a\neq f(q)$ we forbid the pattern $q'_a$ with the domain
$N\times\{0\}\cup\{(\vec{0},1)\}$ where $q'_a(\vec{u},0) = q(\vec{u})$
for all $\vec{u}\in N$ and $q'_a(\vec{0},1)=a$, i.e. consecutive slices are prevented from having an update error
according to the local rule $f$. Let $P'$ be the set of all $\hat{p}$ and $q'_a$. Then clearly $\SP{F}=\XP_{P'}$.

\begin{lemma}
\label{lem:SPalgo}
Given $\X$ and $F$ one can effectively construct $\SP{F}$.
\qed
\end{lemma}

\subsection*{Traces}

Let  $D\subseteq \Z^d$ be finite. For any orbit  $\dots ,c^{(-1)},c^{(0)},c^{(1)},\dots$ the sequence
$\dots ,c^{(-1)}|_{D},c^{(0)}|_{D},$ $c^{(1)}|_{D},\dots$ of consecutive views in the domain $D$
is a \emph{$D$-trace}. Each $c^{(i)}|_{D}$ is an element of the finite group $\G^D$, and hence the trace is a
one-dimensional configuration over the group $\G^D$. Let us denote by $\trace{F}{D}\subseteq (\G^D)^\Z$
the set of all $D$-traces of $F$.

\begin{lemma}
$\trace{F}{D}$ is a one-dimensional group shift over $\G^D$. It is the projection
of $\SP{F}$ on $D\times\Z$.
\qed
\end{lemma}

We call the set $\trace{F}{D}$ the $D$-\emph{trace subshift} of $F$, or simply a trace subshift of $F$. It can be effectively constructed:
Given $\X$ and $F$ we can use
Lemma~\ref{lem:SPalgo} to effectively construct the group shift $\SP{F}$ of space-time diagrams, and then by
Corollary~\ref{cor:projection} we can effectively construct the projection
$\trace{F}{D}$  of $\SP{F}$ on $D\times\Z$.

\begin{lemma}
\label{lem:trace}
Given $\X$ and $F$ and any finite $D\subseteq\Z^d$, one can effectively construct $\trace{F}{D}$.
\qed
\end{lemma}

\subsection*{Limit sets}

The \emph{limit set} $\limit{F}$ of a cellular automaton $F$ consists of all configurations $c^{(0)}\in\X$ that
are present in some bi-infinite orbit $\dots c^{(-1)},c^{(0)},c^{(1)},\dots$
In other words, $\limit{F}$ is the set of the $d$-dimensional slices
of thickness one of $\SP{F}$ in the $d$ spatial dimensions. As a projection of the group shift $\SP{F}$,
the set $\limit{F}$ is a group shift.

\begin{lemma}
$\limit{F}$ is a $d$-dimensional group shift over $\G$. It is the projection
of $\SP{F}$ on $\Z^d\times\{0\}$.
\qed
\end{lemma}

Using Corollary~\ref{cor:projection} we immediately get an algorithm to construct the limit set.

\begin{lemma}
\label{lem:limitset}
Given $\X$ and $F$, one can effectively construct $\limit{F}$.
\qed
\end{lemma}

By definition it is clear that $F(\limit{F})=\limit{F}$ so that $F$ is surjective on its limit set.
By a simple compactness argument we have that $\limit{F}=\bigcap_{n\in\N} F^{n}(\X)$, stating that any configuration that has
arbitrarily long sequences of pre-images has an infinite sequence of pre-images. Note that $\X\supseteq F(\X)\supseteq F^{2}(\X)\supseteq$ is a decreasing chain of group shifts. By Theorem~\ref{thm:nodecreasingchain}
there are no infinite strictly decreasing chains of group shifts, so we have that $F^{k+1}(\X)=F^{k}(\X)$
holds for some $k$. Then $F^{j}(\X)=F^{k}(\X)$ for all $j>k$ so that $\limit{F}=F^{k}(\X)$. So all group cellular automata
reach their limit set after a finite time:
\begin{lemma}
\label{lem:stable}
Group cellular automata $F:\X\longrightarrow\X$ are stable in the sense that there exists
$k\in\N$ such that $F^k(\X)=\limit{F}$.
\qed

\end{lemma}

\subsection*{Periodic points}

A well-known open problem due to Blanchard and Tisseur asks whether every surjective cellular automaton
on a (one-dimensional) full shift has a dense set of temporally periodic points. This has been proved to be the
case in a number of restricted setups, including additive cellular automata on the one-dimensional full shift~\cite{DennunzioTCS2020}.
In fact, Theorem~\ref{thm:periodic} implies the result for all group cellular automata, for any dimension and
on any group shift, not just the full shift. Even jointly periodic configurations are dense:
a configuration is called \emph{jointly periodic} for a cellular automaton if it is temporally periodic and also totally periodic in space.

\begin{corollary}
\label{cor:densetemporallyperiodic}
Let $F:\X\longrightarrow \X$ be a group cellular automaton on a $d$-dimensional group shift $\X$. Jointly periodic configurations are dense in $\limit{F}$. In particular, if $F$ is surjective then they are dense in $\X$.
\end{corollary}

\begin{proof}
By Lemma~\ref{lem:SPisgroupshift} the set $\SP{F}$ of space-time diagrams is a $(d+1)$-dimensional
group shift, and by Theorem~\ref{thm:periodic} totally periodic elements are dense in $\SP{F}$. The projection $\pi(c)$
of a totally periodic space-time diagram $c$ on the domain $\Z^d\times\{0\}$
is a totally periodic element of $\limit{F}$ that is also temporally periodic. The density of totally periodic space-time diagrams $c$ in $\SP{F}$
implies the density of their projections $\pi(c)$ in $\limit{F}=\pi(\SP{F})$. If $F$ is surjective then $\limit{F}=\X$.
\end{proof}

\subsection*{Injectivity and surjectivity}

Another immediate implication of Theorem~\ref{thm:periodic} is a \emph{surjunctivity} property:
every injective group cellular automaton $F:\X\longrightarrow\X$ is surjective.

\begin{corollary}
\label{cor:injectiveimpliessurjective}
Let $F:\X\longrightarrow \X$ be a group cellular automaton on a $d$-dimensional group shift $\X$.
If $F$ is injective then it is surjective.
\end{corollary}

\begin{proof}
If $F$ is injective then it is injective among totally periodic configurations of $\X$. For any fixed $k>0$
there are finitely many configurations in $\X$ that are $k\vec{e}_i$-periodic for all $i\in\{1,\dots ,d\}$.
These are mapped by $F$ injectively to each other. Any injective map on a finite set is also surjective, so
we see that $F$ is surjective among totally periodic configurations of $\X$. By Theorem~\ref{thm:periodic}
the totally periodic configurations are dense in $\X$ so that $F(\X)$ is a dense subset of $\X$.
By the continuity of $F$ it is also closed which means that $F(\X)=\X$.
\end{proof}

We have that every injective group cellular automaton is bijective.
Recall that a bijective cellular automaton $F$ is automatically reversible, meaning that the inverse $F^{-1}$
is also a cellular automaton. If $F$ is a reversible group cellular automaton then clearly so is $F^{-1}$.
Reversible cellular automata are of particular interest due to their relevance in modeling microscopic physics
and in other application domains~\cite{revsurvey}. While it is decidable if a given one-dimensional
cellular automaton is injective (=reversible) or surjective, the same questions are undecidable for general
two-dimensional cellular automata~\cite{karireversible}. As expected, the situation is different for group cellular automata.

\begin{theorem}
\label{thm:injectivity}
It is decidable if a given group cellular automaton $F:\X\longrightarrow \X$ over a given $d$-dimensional group shift $\X$
is injective (surjective).
\end{theorem}

\begin{proof}
By Lemma~\ref{lem:limitset} one can effectively construct the limit set $\limit{F}$. The CA $F$ is surjective if and only if
$\limit{F}=\X$. As equality of given group shifts is decidable (Corollary~\ref{cor:groupequivalnece}(b)), it follows that surjectivity is decidable.

For injectivity, recall that a group homomorphism $F$ is injective if and only if $\kernel{F}=\{\id_{\X}\}$.
Since $\kernel{F}$ is a group shift that can be effectively constructed (Lemma~\ref{lem:kernel}), we can check injectivity by checking
the equality of the two group shifts $\kernel{F}$ and $\{\id_{\X}\}$.
\end{proof}

\subsection*{The Garden-of-Eden-theorem}

The Garden-of-Eden-theorem is among the oldest results in the theory of cellular automata. It links injectivity and surjectivity. Let us call two configurations $c,e\in A^{\Z^d}$
\emph{asymptotic}  if their \emph{difference set} $\diff{c}{e}=\{\vec{u}\in\Z^d\ |\ c_{\vec{u}}\neq e_{\vec{u}}\}$ is finite. A cellular automaton $F:X\longrightarrow X$ on a subshift $X$
is called \emph{pre-injective} if for any asymptotic $c,e$ the following holds: $c\neq e\Longrightarrow F(c)\neq F(e)$. So  injectivity is only required among mutually asymptotic configurations.
Trivially every injective cellular automaton is pre-injective but the converse implication is not true. In fact, the classical
Garden-of-Eden-theorem states that on full shifts in any dimension pre-injectivity is equivalent to surjectivity.

\begin{theorem}[the Garden-of-Eden-theorem~\cite{moore,myhill}]
A cellular automaton $F:A^{\Z^d}\longrightarrow A^{\Z^d}$ is pre-injective if and only if it is surjective.
\end{theorem}

That surjectivity implies pre-injectivity was first proved by E.F.Moore~\cite{moore}, and the converse implication a year later by J.Myhill~\cite{myhill}. Later the theorem has been
extended to many other settings. For example, it is known that the Garden-of-Eden-theorem holds for
cellular automata over so-called \emph{strongly irreducible} subshifts of finite type~\cite{fiorenzi}.

Note that the Myhill direction implies surjunctivity: if a cellular automaton is injective then it is pre-injective and by Myhill's theorem surjective.
For group shifts we proved surjunctivity differently in Corollary~\ref{cor:injectiveimpliessurjective}, using the density of periodic points.
There is a good reason for this: the Myhill direction of the Garden-of-Eden-theorem is namely
not true for all group cellular automata over group shifts, as shown by the following trivial example.

\begin{example}
Let $\X=\{0^{\Z}, 1^{\Z}\}$ be the two-element group shift over the two-element cyclic group $\Z_2$, and let $F:\X\longrightarrow \X$  be the group
cellular automaton $F(0^{\Z})=F(1^{\Z})=0^{\Z}$. Then $F$ is pre-injective but not surjective.
\end{example}

Recall that it is decidable whether a given group cellular automaton is surjective
(Theorem~\ref{thm:injectivity}). Since surjectivity and pre-injectivity are not equivalent for all group cellular automata, a natural
follow up question is to determine if a given group cellular automaton is pre-injective. The decidability status of this question remains open.

\begin{question}
Is it decidable if a given group cellular automaton is pre-injective ?
\end{question}

Next we show that the Moore direction of the Garden-of-Eden-theorem holds for all group cellular automata.
The proof is based on the fact that all
surjective cellular automata preserve entropy
while group cellular automata that are not pre-injective do not
preserve it. The topological \emph{entropy}  of a $d$-dimensional subshift $X$ is defined as $$h(X)=\lim_{n\to\infty}\frac{\log |\Patt{X}{B_n}| }{|B_n|}$$
where $B_n=\{1,\dots ,n\}^d$  is the $d$-dimensional box  of size $n\times \dots \times n$.  The limit exists by Fekete's Subadditive Lemma.

\begin{theorem}
\label{thm:group_goe}
Let $F:\X\longrightarrow \X$ be a group cellular automaton over a group shift $\X\subseteq\G^{\Z^d}$. If $F$ is surjective then $F$ is pre-injective.
\end{theorem}

\begin{proof}
Suppose $F$ is not pre-injective so $F(x)=F(y)$ for an asymptotic pair $x,y\in\X$, $x\neq y$. Then $c=xy^{-1}\in\X$ is
asymptotic with $\id_{\X}$ while $c\neq \id_{\X}$ and $F(c)=F(\id_{\X})=\id_{\X}$.
It follows from this fact that the entropy of the kernel of $F$ is strictly positive, $h(\kernel{F})>0$. However, one can easily prove using the first isomorphism theorem of groups that
for the entropies of the group shifts $\X$, $F(\X)$ and $\kernel{F}$ the following addition formula holds:
$h(\X)=h(F(\X))+h(\kernel{F})$. We then have that $h(\X)>h(F(\X))$, implying that $\X\neq F(\X)$, \emph{i.e.}, that $F$ is not surjective.
\end{proof}

\subsection*{Nilpotency, equicontinuity and sensitivity}

A cellular automaton is called \emph{nilpotent} if there is only one configuration in the limit set $\limit{F}$.
(Clearly the limit set is never empty.) Nilpotency is undecidable even for cellular automata over one-dimensional full shifts~\cite{kari1992nilpotency,Lewis79}. In the case of group cellular automata the identity configuration is a fixed point and hence automatically in the limit set. Nilpotency
of group cellular automata can be easily tested by effectively constructing the limit set (Lemma~\ref{lem:limitset})
and testing equivalence with the singleton group shift $\{\id_{\X}\}$.

More generally, a cellular automaton $F$ is \emph{eventually periodic} if $F^{n+p}=F^n$ for some $n$ and  $p\geq 1$, and
it is \emph{periodic} if $F^p$ is the identity map for some $p\geq 1$.
Nilpotent cellular automata are clearly eventually periodic with $p=1$. Note that eventually periodic
cellular automata are periodic on the limit set and, conversely, if $F$ is periodic on its limit set then
it is eventually periodic on $\X$ because $\limit{F}=F^n(\X)$ for some $n$ by Lemma~\ref{lem:stable}.

\begin{theorem}
\label{thm:decideperiodicity}
 It is decidable for a given group cellular automaton $G:\X\longrightarrow\X$ on a given $d$-dimensional
 group shift $\X$ whether $F$ is nilpotent, periodic or eventually periodic.
\end{theorem}

\begin{proof}
We have that $F$ is
\begin{itemize}
\item nilpotent if and only if $\limit{F}=\{\id_{\X}\}$,
\item eventually periodic if and only if $\trace{F}{\{\vec{0}\}}$ is finite,
\item periodic if and only if it is injective and eventually periodic.
\end{itemize}
Group shifts $\limit{F}$ and
$\trace{F}{\{\vec{0}\}}$ can be effectively constructed (Lemma~\ref{lem:trace} and Lemma~\ref{lem:limitset}).
Equivalence of $\limit{F}$ and $\{\id_{\X}\}$ can be tested (Corollary~\ref{cor:groupequivalnece}(b))
and finiteness of a given one-dimensional subshift of finite type is easily checked, so nilpotency and eventual periodicity are decidable. By Theorem~\ref{thm:injectivity} injectivity of $F$ is decidable so also periodicity can be decided.
\end{proof}

A configuration $c\in\X$ is an \emph{equicontinuity point} of $F:\X\longrightarrow\X$ if for every finite $D\subseteq\Z^d$
there exists a finite $E\subseteq\Z^d$ such that $e|_{E}=c|_{E}$ implies $F^n(e)|_{D}=F^n(c)|_{D}$
for all $n\geq 0$. Orbits of equicontinuity points can hence
be reliably simulated even if the initial configuration is
not precisely known. Let $\eq{F}\subseteq\X$ be the set of equicontinuity points of $F$. We call $F$ \emph{equicontinuous} if
$\eq{F} = \X$.

Cellular automaton $F:\X\longrightarrow\X$ is \emph{sensitive to initial conditions}, or just \emph{sensitive}, if
there exists a finite observation window $D\subseteq\Z^d$ such that for every configuration $c\in\X$ and every
finite $E\subseteq\Z^d$ there is $e\in\X$ with $e|_{E}=c|_{E}$ but $F^n(e)|_{D}\neq F^n(c)|_{D}$ for some $n\geq 0$.
Clearly $c$ cannot be an equicontinuity point so for all sensitive $F$ we have $\eq{F} = \emptyset$. For group cellular automata
also the converse holds.

\begin{lemma}
\label{lem:dichotomy}
Let $F:\X\longrightarrow \X$ be a group cellular automaton over a $d$-dimensional group shift $\X$.
Then exactly one of the following two possibilities holds:
\begin{itemize}
\item $\eq{F} = \X$ and $F$ is equicontinuous, or
\item $\eq{F} = \emptyset$ and $F$ is sensitive.
\end{itemize}
\end{lemma}

\begin{proof}
Assume that some $c\not\in\eq{F}$ exists, which means that
there exists a finite $D\subseteq\Z^d$ such that for all finite
$E\subseteq\Z^d$ there is $e\in\X$ and $n\geq 1$ with $e|_{E}=c|_{E}$ but $F^n(e)|_{D}\neq F^n(c)|_{D}$.
Consider an arbitrary $c'\in\X$.
For $c''=c'ec^{-1}\in \X$ we then have that
$c''|_{E}=c'|_{E}$ but $F^n(c'')|_{D}\neq F^n(c')|_{D}$.
This proves that $c'\not\in\eq{F}$.

We can conclude that for group cellular automata either $\eq{F} = \X$ or $\eq{F} = \emptyset$.
By definition, $\eq{F} = \X$ is equivalent to equicontinuity of $F$.

If $F$ is sensitive then $\eq{F} = \emptyset$ holds. Conversely, if $F$ is not sensitive then, by definition,
for all finite $D\subseteq\Z^d$ there exists $c\in\X$ and a finite $E\subseteq\Z^d$ such that $e|_{E}=c|_{E}$
implies that $F^n(e)|_{D}= F^n(c)|_{D}$ for all $n\geq 0$.
As above, we can replace $c$ by any other configuration $c'$, which
implies that all configurations are equicontinuity points, i.e., $\eq{F} \neq \emptyset$.
\end{proof}

We can decide equicontinuity and sensitivity.

\begin{theorem}
\label{thm:sensitive}
 It is decidable for a given group cellular automaton $G:\X\longrightarrow\X$ on a given $d$-dimensional
 group shift $\X$ whether $F$ is equicontinuous or sensitive to initial conditions.
\end{theorem}

\begin{proof}
By the dichotomy in Lemma~\ref{lem:dichotomy} it is enough to decide equicontinuity. Let us show that $F$ is equicontinuous if
and only if it is eventually periodic, after which the decidability follows from Theorem~\ref{thm:decideperiodicity}.

If $F$ is eventually periodic then it is trivially equicontinuous since there are only finitely many different
functions $F^k$, $k\geq 0$, and all these functions are continuous.  Conversely, if $F$ is equicontinuous
then one easily sees that there are only finitely many different traces in $\trace{F}{\{\vec{0}\}}$.
Indeed, equicontinuity at configuration $c$ implies that there is a finite set
$E\subseteq\Z^d$ such that $e|_{E}=c|_{E}$ implies that $F^n(e)_{\vec{0}}=F^n(c)_{\vec{0}}$
for all $n\geq 0$. As in the proof of Lemma~\ref{lem:dichotomy} we see that the same set $E$ works for all
configurations $c$. But then $|\Patt{\X}{E}|$ is an upper bound on the number of
different traces in $\trace{F}{\{\vec{0}\}}$ because
$c|_{E}$ uniquely identifies the
positive trace of $c$ (and by the translation invariance of the trace subshift any $k$ different traces
can be shifted to provide $k$ different positive traces.)

Finiteness of $\trace{F}{\{\vec{0}\}}$ implies that all traces are periodic with a common period, so that
cellular automaton $F$ is periodic on its limit set. Hence $F$ is eventually periodic.
\end{proof}

\subsection*{Mixing properties}

A cellular automaton $F:\X\longrightarrow \X$ is \emph{transitive} if there is an orbit from every non-empty open set to
every non-empty open set, that is, if for any finite $D\subseteq\Z^d$ and all $p,q\in \Patt{\X}{D}$
there exists $c\in\X$ and $n\geq 0$ such that $c|_{D}=p$ and $G^n(c)|_{D}=q$. It is \emph{mixing} if there exists such $c$ for every sufficiently large $n$, that is, if for all $D,p$ and $q$ as above there is $m$ such that for every $n\geq m$
there exists $c\in\X$  such that $c|_{D}=p$ and $G^n(c)|_{D}=q$.

For these properties we obtain only semi-algorithms for the negative instances. Decidability remains open.

\begin{theorem}
\label{thm:mixing}
 It is semi-decidable for a given group cellular automaton $G:\X\longrightarrow\X$ on a given $d$-dimensional
 group shift $\X$ whether $F$ is non-transitive or non-mixing.
\end{theorem}
\begin{proof}
A non-deterministic semi-algorithm guesses a finite $D\subseteq \Z^d$, forms the trace subshift $\trace{F}{D}$,
and verifies that the trace subshift is not transitive (not mixing, respectively). Clearly $F$ is not transitive (not mixing, respectively) if and only if such a choice of $D$ exists. For one-dimensional subshifts of finite type,
such as $\trace{F}{D}$, it is easy to decide transitivity and the mixing property~\cite{LindMarcus95}.
\end{proof}

\begin{question}
Is it decidable if a given group cellular automaton is transitive (or mixing) ?
\end{question}

\section{Conclusions}
\label{sec:conclusions}

We have demonstrated how the ``swamp of undecidability''~\cite{lindswamp} of multidimensional SFTs and cellular automata
is mostly absent in the group setting. For general cellular automata nilpotency~\cite{kari1992nilpotency,Lewis79}, as well as eventual periodicity, equicontinuity and sensitivity~\cite{equicontinuity} are undecidable on one-dimensional full shifts, and
periodicity~\cite{periodicity}, as well as sensitivity, mixingness and transitivity~\cite{lukkarila} are undecidable
even among reversible one-dimensional cellular automata on the full shift;
injectivity and surjectivity are undecidable for two-dimensional cellular automata on the full shift~\cite{karireversible}.
Algorithms and characterizations
have been known for linear and additive cellular automata (on full shifts, sometimes depending on the dimension~\cite{DennunzioTCS2020,DennunzioPaper2019}). Our results improve these  to the
greater generality of non-commutative groups and cellular automata on higher dimensional subshifts.
However, it should be noted that the
existing characterizations in the literature typically provide easy to check conditions on the
local rule of the cellular automaton for the considered properties, while algorithms extracted from our proofs are
impractical and only serve the purpose of proving decidability.

It remains open whether it is decidable if a given group cellular automaton is pre-injective, transitive or mixing.

\bibliographystyle{unsrt}
\bibliography{additive_SI}

\end{document}